\documentclass[12pt]{iopart}
\usepackage{iopams}
\usepackage[backend=biber, style=numeric-comp, sorting=none]{biblatex}
\usepackage{amsthm}
\usepackage{amssymb}
\usepackage{amsfonts}
\usepackage{graphicx}
\usepackage{parskip}
\usepackage{breakurl}
\usepackage{hyperref}
\hypersetup{colorlinks=true, allcolors=blue, breaklinks=true}
\usepackage{changepage}
\usepackage{enumitem}
\newtheorem{thm}{Theorem}
\newtheorem{df}{Definition}

\newtheorem{cor}{Corollary}
\newcommand{\s}[0]{\textnormal{ }}
\newcommand{\eqva}[0]{\Leftrightarrow}
\eqnobysec

\begin{document}
\title{General characterisation of Hamiltonians generating velocity-independent forces}
\author{F Yip$^1$, A C H Cheung$^{2,3}$}
\address{$^1$ Trinity College, University of Cambridge,
Cambridge CB2~1TQ, UK}
\address{$^2$ Theory of Condensed Matter Group, Cavendish Laboratory, University of Cambridge, J.~J.~Thomson Avenue, Cambridge, CB3 0HE, UK}
\address{$^3$ The London Interdisciplinary School, 20-30 Whitechapel Road, London, E1 1EW, UK}
\eads{\mailto{fy276@cam.ac.uk}, \mailto{achc2@cam.ac.uk}}

\begin{abstract}
    Dynamics generated from Hamiltonians enjoy potential pathways to quantisation, but standard Hamiltonians are only capable of generating conservative forces. Classes of Hamiltonians have been proposed in Berry \emph{et al}\cite{Berry_2015}.~capable of generating non-conservative velocity-independent forces. Such Hamiltonians have been classified in the past, under the strict assumption that they are polynomial in momentum. This assumption is relaxed here to analyticity. In doing so, broader classes of Hamiltonians are discovered. 
    
    By considering the Hamiltonian as a function of state space without introducing the Lagrangian and constructing a metric-like tensor, we develop strong general constraints on Hamiltonians generating velocity-independent forces and exhibit a surprising dichotomy between classes of such Hamiltonians. These results are applicable to any spatial domain of any dimension admitting well-defined Hamiltonian dynamics. As an example application, we apply these constraints to classify all Hamiltonian velocity-independent forces in two spatial dimensions, as well as all such Hamiltonians which do not generate an isotropic simple harmonic motion. The case of one spatial dimension is also discussed for the sake of completeness. 
\end{abstract}
\noindent{\it hamiltonian characterisation, curl forces, hamiltonian dynamics, classical dynamics\/}

\submitto{\jpa}

\section{Introduction}

We consider the motion of a classical particle in $\mathbb{R}^n$ governed by a time-independent Hamiltonian \cite{Arnold}. Scaling its mass to unity, we identify force and acceleration via Newton's Second Law, instead of considering force as the time derivative of momentum. 

We investigate which Hamiltonians are capable of generating Newtonian forces of the form

\begin{equation} \label{intro Newton}
    \ddot{\mathbf{x}} = \mathbf{F}(\mathbf{x}), 
\end{equation}

which are functions of position $\mathbf{x}$ and independent of velocity $\mathbf{v}$. 

On one hand, standard Hamiltonians of the form 
\begin{equation}
    H = \frac{1}{2}\mathbf{p^{T}} \mathbf{p} + U(\mathbf{x})
\end{equation}
generate conservative forces \cite{Arnold} (where $\mathbf{p}$ denotes momentum). On the other hand, a fully general\footnote{This work considers time-independent Hamiltonians. The physical systems motivating this investigation are often inherently time-independent, being results of the time-averaging of fast, driven modes. The appendix offers an initial brief outlook into the \textit{time-dependent} setting. } Hamiltonian 
\begin{equation}
    H = H(\mathbf{x}, \mathbf{p})
\end{equation}
generates forces with arbitrary dependence on momentum $\mathbf{p}$, hence losing the interpretability as a model of the Newtonian dynamics (\ref{intro Newton}). 

These velocity-independent forces which are not conservative have been known as `curl forces' \cite{Berry_2012, Berry_2013, Berry_2015, Berry_2016, Berry_2020}, `follower forces' \cite{FollowerForce1, FollowerForce2, FollowerForce3} and `circulatory forces' \cite{CirculatoryForce1, CirculatoryForce2, CirculatoryForce3, CirculatoryForce4, CirculatoryForce5}. Whilst all fundamental forces are conservative, curl forces occur as effective forces in situations such as the dynamics of dipoles in optical electromagnetic fields \cite{optical1, optical2, optical3, optical4, optical5, optical6, optical7, optical8, Berry_2013}, whirling flexible shafts \cite{WhirlingShaftKimbal, WhirlingShaftSmithStoney} and pendulums with rotational dissipation \cite{Crandall1995}. 

The interest in modelling the Newtonian dynamics (\ref{intro Newton}) by a Hamiltonian arises from both the partial integrability \cite{Arnold} and more natural potential paths to quantisation \cite{Berry_2020} offered by Hamiltonian dynamics. 

In section \ref{highD}, we develop strong general results on Hamiltonians over any open and connected spatial domain generating velocity-independent forces, subject to physically reasonable regularity conditions described in subsection \ref{regularity}. 

We apply these results to the case of one and two spatial dimensions in section \ref{1D prob sec} and section \ref{2D prob sec}. A complete classification of such Hamiltonians in 1D and a partial classification of such Hamiltonians in 2D are obtained, in which new possibilities emerge beyond the previously studied case where the Hamiltonian is polynomial in momentum \cite{Berry_2015}. We note a curious connection to the Seesaw Mechanism \cite{SeesawMinkowski, SeesawMassMatrix1, SeesawTsutomu, SeesawMassMatrix2} in one case. We consequently obtain also a classification of all velocity-independent Hamiltonian forces in 2D. (In 1D all velocity-independent forces are conservative and hence Hamiltonian.) 

We discuss the relation between velocity-independent Hamiltonian forces and conservative forces in section \ref{pseudo conservative}, where we introduce the concepts of pseudo conservative and pseudo curl-free forces. 

\section{Assumptions and Notations}
\subsection{Regularity Assumptions} \label{regularity}
We consider Hamiltonians on the phase space $(\mathbf{x}, \mathbf{p})\in D\times \mathbb{R}^n$, where $D\subseteq \mathbb{R}^n$ is any open and connected subset. We do not assume that $D\times \mathbb{R}^n$ is closed under the evolution generated by the Hamiltonian. In other words, trajectories are allowed to leave $D$, at which point they terminate. 

All general results developed in section \ref{highD} will in fact be applicable to any (open and connected) spatial domain on which Hamiltonian dynamics is well-defined, such as (open and connected) subsets of a cylinder or a torus. 

As the spatial domain $D$ is not forced to be $\mathbb{R}^n$, the results are applicable to Hamiltonians with spatial singularities such as, but not limited to, a finite number of singularities, a lattice of singularities or a branch cut. Our motivation for allowing these singularities is to allow the generation of forces with (spatial) singularities. 

In our treatment here, we assume always, implicit in the term `Hamiltonian', that a bijection between momentum $\mathbf{p}$ and velocity $\mathbf{v}$ is generated at every $\mathbf{x}\in D$. 

We assume also that $H$ is analytic as a function of $(\mathbf{x}, \mathbf{p})$ and the bijection between $\mathbf{p}$ and $\mathbf{v}$ is an analytic diffeomorphism. This is physically reasonable as it is equivalent to the assumption that $H$ is an analytic function of $(\mathbf{x}, \mathbf{p})$ and the Lagrangian $L$ is an analytic function of $(\mathbf{x}, \mathbf{v})$, 

We refer to such Hamiltonians as \textit{regular}, all Hamiltonians are assumed to be regular unless otherwise stated. 

\subsection{Canonical Transformations} \label{canon trans sec}

We call a canonical transformation \cite{Arnold} \textit{permissible} if it preserves the velocity-independence of the force. The following two types of canonical transformations are permissible. When combined they generate the most general permissible canonical transformation. 

\begin{itemize}
    \item Type 1: 
    \begin{equation}
        \mathbf{x} \rightarrow N\mathbf{x} + \mathbf{b},\qquad\mathbf{p} \rightarrow (N^T)^{-1}\mathbf{p}
    \end{equation}
    for a fixed invertible matrix $N$ and vector $\mathbf{b}$. 
    \item Type 2: 
    \begin{equation}
        \mathbf{x} \rightarrow \mathbf{x},\qquad\mathbf{p} \rightarrow \mathbf{p} + \mathbf{V}(\mathbf{x})
    \end{equation}
    for any curl-free vector field $\mathbf{V}$ of $\mathbf{x}$. 
\end{itemize}

Hamiltonians related by a canonical transformation give rise to the same quantum system. As such, we shall only classify Hamiltonians up to these canonical transformations and make liberal use of them throughout. 

\subsection{Index Positions} \label{idx pos sec}

We shall adopt the index positions $x^i$ and $p_i$. The use of opposite index positions is justified by canonical transformations of type 1. (These shall be viewed as tensors over $GL_{\textnormal{dim}}(\mathbb{R})$, not $SO_{\textnormal{dim}}(\mathbb{R})$.)

Thus $\mathbf{x}$ derivatives should be of the form $\partial_i$ and $\mathbf{p}$ derivatives should be $\partial^i$. 

Due to the bijection between phase space ($\mathbf{x}, \mathbf{p}$) and state space ($\mathbf{x}, \mathbf{v}$), functions on phase space are equivalently functions of state space. Thus we can take $\mathbf{v}$ derivatives too. 

Specifically, we denote the $\mathbf{x}$ derivative whilst fixing $\mathbf{p}$ as $\partial'_i$, and the $\mathbf{x}$ derivative whilst fixing the velocity $\mathbf{v}$ as $\partial_i$. Note that if we are differentiating a function of $\mathbf{x}$ only, then these distinctions are not needed. 

Written in index notation with summation convention, the two types of permissible canonical transformations are: 

\begin{itemize}
    \item Type 1: 
    \begin{equation}
        x^i \rightarrow N^i_{\s j} x^j + b^i,\qquad p_i \rightarrow (N^{-1})^j_{\s i} p_j
    \end{equation} 
    for a fixed invertible matrix $N$ and vector $\mathbf{b}$. 
    \item Type 2: 
    \begin{equation}
        x_i \rightarrow x_i,\qquad p_i \rightarrow p_i + V_i(\mathbf{x})
    \end{equation}
    such that $\partial_iV_j$ is symmetric. 
\end{itemize}

By definition, these preserve Poisson brackets \cite{Arnold} and transform the force $F^i$ as a tensor. 

\subsection{Differential Operators}

We study now the effect of these canonical transformations on differential operators. 

Canonical transformations of type 1 transform $\partial_i$, $\partial'_i$ and $\partial^i$ as tensors. 

Canonical transformations of type 2 fix $\partial_i$ and $\partial^i$ but not $\partial'_i$. 

Such information is crucial, so that we know what we define later on will be meaningful and independent of these canonical transformations. 

\section{General Results} \label{highD}

Noting that due to the analytic diffeomorphism between phase space $(\mathbf{x}, \mathbf{p})$ and state space $(\mathbf{x}, \mathbf{v})$, any analytic function of one is also an analytic function of the other. Our discussion will largely involve translating the problem to the state space $(\mathbf{x}, \mathbf{v})$, where we solve the problem before translating back to phase space $(\mathbf{x}, \mathbf{p})$. 

We do not establish the Lagrangian, instead we simply write the Hamiltonian as a function of $(\mathbf{x}, \mathbf{v})$. Indeed, the bijection between phase and state space may be recovered from this via some of the properties we will prove. In doing so, we may recover the Hamiltonian as a function of phase space $(\mathbf{x}, \mathbf{p})$ and the dynamics. However, we will not need to do so explicitly here. 

As we are scaling the mass of the particle to $1$, the force is given by

\begin{equation}
    F^i = \dot{v^i} = \{\partial^i H, H\}, 
\end{equation}

where $\{f, g\} = \partial'_i (f) \partial^i (g) - \partial'_i (g) \partial^i (f)$ is the Poisson bracket \cite{Arnold}. Taking a momentum derivative gives

\begin{equation}
    \{\partial^i H, \partial^jH\} + \{\partial^{ij} H, H\} = 0
\end{equation}

by velocity-independence. 

Note that the first term is antisymmetric in $i, j$ (as the Poisson bracket is antisymmetric) and the second term is symmetric in $i, j$. Hence separating by symmetry, this is equivalent to

\begin{equation} \label{multiD fund1}
    \{\partial^i H, \partial^jH\} = 0,
\end{equation}

and

\begin{equation} \label{multiD fund2}
    \{\partial^{ij} H, H\} = 0.
\end{equation}

\subsection{Introducing \texorpdfstring{$g_{ij}$}{}}

We define $g^{ij} = \partial^{ij} H$. Note that this definition is invariant under canonical transformations of type 2 and hence meaningful. 

This is the Hessian for the $\mathbf{p} - \mathbf{v}$ bijection. Given that $H$ is an analytic function of $(\mathbf{x}, \mathbf{p})$, the $\mathbf{p} - \mathbf{v}$ bijection being an analytic diffeomorphism is equivalent to $g^{ij}$ being always non-singular. Let $g_{ij}$ denote the inverse of $g^{ij}$. 

As the Hessian, we have that
\begin{eqnarray}
    \partial^i = g^{ij}\partial_{v^j}, \\
    \partial_{v^i} = g_{ij}\partial^j. 
\end{eqnarray}

$g_{ij}$ plays a role similar to a spatial metric. Indeed for the standard Hamiltonian 
\begin{equation}
    H = \frac{1}{2}\mathbf{p^{T}} \mathbf{p} + U(\mathbf{x}), 
\end{equation}
$g_{ij}$ coincides with the standard spatial metric $\delta_{ij}$. 

\subsection{Coinciding Constant \texorpdfstring{$\mathbf{v}$}{} and \texorpdfstring{$\mathbf{p}$}{}}

We claim that for any choice of $\mathbf{v_0}$, we can perform a canonical transformation of type 2 so that $\mathbf{v} = \mathbf{v_0} \eqva \mathbf{\tilde{p}} = \mathbf{0}$, where $\mathbf{\tilde{p}}$ is the momentum after the canonical transformation. 

Indeed let $\mathbf{u} (\mathbf{x})$ denote the momentum originally for which $\mathbf{v} = \mathbf{v_0}$. 

Taking $\mathbf{p} \rightarrow \mathbf{p} - \mathbf{u} (\mathbf{x})$ works so long as we can prove that $\mathbf{u}$ is curl-free. 

In the original momentum, we have

\begin{equation}
    v_0^i = \partial^i H (\mathbf{x}, \mathbf{u} (\mathbf{x})). 
\end{equation}

Taking an $\mathbf{x}$ derivative at $(\mathbf{x}, \mathbf{u} (\mathbf{x}))$ gives

\begin{equation}
    0 = \partial'_{j}\partial^{i} H + g^{ik} (\partial_j u_k). 
\end{equation}

On the other hand, (\ref{multiD fund1}) gives 

\begin{equation} \label{multiD fund1 in g}
    g^{jk}\partial'_k\partial^i H \textnormal{ is symmetric}. 
\end{equation}

Thus we have that

\begin{equation}
    g^{ik} g^{jl} (\partial_j u_k) = -g^{jl}\partial'_{j}\partial^{i} H \textnormal{ is symmetric}. 
\end{equation}

As $g$ is invertible, we have that $\partial_j u_k$ is symmetric and $\mathbf{u}$ is curl-free. 

\subsection{5 Basic Properties of \texorpdfstring{$g_{ij}$}{}}

\subsubsection{Property 1}
We have

\begin{equation} \label{g basic 1}
    \eqalign{\partial_{v^i}g_{jk} &= g_{ia}\partial^a g_{jk}\\
    &= -g_{ia}g_{jb}g_{kc}\partial^a g^{bc}\\
    &= -g_{ia}g_{jb}g_{kc}\partial^{abc} H. }
\end{equation}

Thus $\partial_{v^i}g_{jk}$ is symmetric. 

\subsubsection{Property 2}
We have from (\ref{multiD fund2}) that $g^{ij}$ is preserved under time evolution, and hence so is $g_{ij}$. Thus we have 

\begin{equation} \label{g basic 2}
    \partial_i(g_{jk})v^i + \partial_{v^i}(g_{jk})F^i = 0. 
\end{equation}

\subsubsection{Property 3}
We have a pair of properties here. 

Applying $\partial_{v^i} = g_{ij}\partial^j$ to $H$ gives

\begin{equation} \label{g basic 3'}
    \partial_{v^i}H = g_{ij}v^j. 
\end{equation}

At any given point ($\mathbf{x}, \mathbf{v_0}$) in state space, we can perform the canonical transformation of type 2 to shift the momentum so that it is $\mathbf{0}$ whenever $\mathbf{v} = \mathbf{v_0}$. So at ($\mathbf{x}, \mathbf{v}$), $\partial_i$ coincides with $\partial'_i$. 

From $g_{ij}$ being the Hessian of the $\mathbf{v} - \mathbf{p}$ map and Hamilton's equation it follows that

\begin{equation} \label{g basic 3}
    \partial_i H = - g_{ij}F^j. 
\end{equation}

Note that both sides of the equation are invariant under canonical transformations of type 2, so the equation holds irrespective of the canonical transformation used to obtain it. 

\subsubsection{Property 4}
Still with $\mathbf{p}$ shifted as above, so $\mathbf{p} = 0$ coincides with $\mathbf{v} = \mathbf{v_0}$, we have by virtue of $g_{ij}$ being the Hessian that

\begin{equation}
    v^i = v_0^i + g^{ij}p_j + \Or(p^2)
\end{equation}

near ($\mathbf{x}, \mathbf{v_0}$). 

We apply (\ref{multiD fund1 in g}) near ($\mathbf{x}, \mathbf{v_0}$) but not necessarily at ($\mathbf{x}, \mathbf{v_0}$), thus we cannot identify $\partial'_i$ with $\partial_i$, giving 

\begin{equation} 
    \partial'_k\partial^i (H)g^{jk} = \partial'_k(v^i)g^{jk} = p_l\partial'_k(g^{il})g^{jk} + \Or(p^2) \textnormal{ is symmetric}
\end{equation}

for any choice of $\mathbf{p}$. 

Hence we have, now at ($\mathbf{x}, \mathbf{v_0}$), and thus identifying $\partial'_i$ with $\partial_i$ that

\begin{equation} 
    \partial_k(g^{il})g^{jk} \textnormal{ is symmetric in } i, j. 
\end{equation}

Thus we have

\begin{equation} 
    \partial_k(g_{IL})g^{iI}g^{lL}g^{jk} \textnormal{ is symmetric in } i, j, 
\end{equation}

and since $g^{ij}$ is invertible, we have

\begin{equation} \label{g basic 4}
    \partial_i(g_{jk}) \textnormal{ is symmetric}. 
\end{equation}

Again, this is invariant under canonical transformations of type 2, so this holds irrespective of the canonical transformation used to obtain it. 

\subsubsection{Property 5}

Taking the $\mathbf{x}$ derivative (fixing $\mathbf{v}$) of property 3 (\ref{g basic 3}) gives 

\begin{equation}
    \partial_{ij} H = - \partial_{j}(g_{ik})F^k - g_{ik}\partial_j(F^k). 
\end{equation}

Let $T_i^j = \partial_i(F^j)$. Noting that $\partial_{ij} H$ and $\partial_{j}(g_{ik})F^k$ are both symmetric, we have that

\begin{equation} \label{crucial eq}
    T^k_j(\mathbf{x})g_{ik}(\mathbf{x}, \mathbf{v}) \textnormal{ is symmetric}. 
\end{equation}

As made explicit with the dependencies of $g$ and $T$, what is remarkable about this is that $T$ is a function of $\mathbf{x}$ only, while $g$ is not in general. We will elevate this to even greater generality in corollary \ref{crucial force}. 

\subsection{A Crucial Property}

\begin{df}
    For $\mathbf{x}\in D$, let $\mathcal{A}(\mathbf{x})$ denote the affine subspace (of $(0, 2)$ tensors) generated by $g_{ij}(\mathbf{x}, \mathbf{v})$ as $\mathbf{v}$ ranges over $\mathbb{R}^n$. 
\end{df}

\begin{thm} \label{crucial}
    If $H$ is regular and generates a velocity-independent force, then $\mathcal{A}(\mathbf{x})$ is independent of $\mathbf{x}\in D$, and as such we will simply denote it as $\mathcal{A}$. (We do not need to assume that the force is not an isotropic SHM.)
\end{thm}

\begin{proof}
    It suffices to prove that for any $\mathbf{x}, \mathbf{y}\in D$, $\mathcal{A}(\mathbf{x}) \subseteq \mathcal{A}(\mathbf{y})$. 

    For this it suffices to show that if $u^{ij}(g_{ij}(\mathbf{y}, \mathbf{v}) - b_{ij}) = 0$ for any $\mathbf{v}$, then $u^{ij}(g_{ij}(\mathbf{x}, \mathbf{v}) - b_{ij}) = 0$ for any $\mathbf{x}, \mathbf{v}$, where $u, b$ are fixed. 

    By analyticity we just need to show that at $\mathbf{y}$, we have 
    
    \begin{equation} \label{powerful diff form}
        u^{ij}\partial_{l_1\cdots l_n}(g_{ij} - b_{ij}) = 0
    \end{equation}
    
    for any $n\in\mathbb{Z}^{\geq 0}$, $l_1, \cdots, l_n$. 

    We prove this by strong induction on $n\in\mathbb{Z}^{\geq 0}$. 

    The base case $n = 0$ is true by assumption. 

    For the inductive case assume that (\ref{powerful diff form}) holds for $0, \cdots, n$ and we seek to prove that it holds for $n+1$, where $n\geq 0$. 

    Acting $\partial_{l_1\cdots l_n}$ on property 2 (\ref{g basic 2}) in the form of $\partial_{l}(g_{ij})v^{l} = - \partial_{v^{l}}(g_{ij})F^{l}$ gives

    \begin{equation}
        \partial_{l_1\cdots l_n l}(g_{ij} - b_{ij})v^{l} = - \partial_{l_1\cdots l_n}(\partial_{v^{l}}(g_{ij} - b_{ij})F^{l}),
    \end{equation}

    where we can add in the $b_{ij}$ as it is fixed and we are taking at least one derivative. Contracting with $u$ gives 
    
    \begin{equation}
        u^{ij}\partial_{l_1\cdots l_n l}(g_{ij} - b_{ij})v^{l} = - u^{ij}\partial_{l_1\cdots l_n}(\partial_{v^{l}}(g_{ij} - b_{ij})F^{l}). 
    \end{equation}

    Note that the right hand side vanishes by the inductive hypothesis as there are at most $n$ positional derivatives acting on $g_{ij} - b_{ij}$. Hence the left hand side vanishes too. 

    Let $(f_{l_1\cdots l_n})_l = u^{ij}\partial_{l_1\cdots l_n l}(g_{ij} - b_{ij})$, where the notation on the left hand side is used to emphasis the index $l$. 

    Thus we have 

    \begin{equation}
        (f_{l_1\cdots l_n})_lv^{l} = 0. 
    \end{equation}

    It follows from Properties 1 and 4 (\ref{g basic 1}, \ref{g basic 4}) that $\partial_{v^k l}g_{ij} = \partial_{v^i l}g_{jk}$ is symmetric in $k, l$, giving

    \begin{equation}
        \partial_{v^k}(f_{l_1\cdots l_n})_l \textnormal{ is symmetric in } k, l. 
    \end{equation}

    Hence we have 
    
    \begin{equation}
        \eqalign{&v^l \partial_{v^l}(v^\mu(f_{l_1\cdots l_n})_{l_{n+1}})\\
        =&v^\mu(f_{l_1\cdots l_n})_{l_{n+1}}+v^lv^\mu \partial_{v^l}((f_{l_1\cdots l_n})_{l_{n+1}})\\
        =&v^\mu(f_{l_1\cdots l_n})_{l_{n+1}}+v^lv^\mu \partial_{v^{l_{n+1}}}((f_{l_1\cdots l_n})_{l})\\
        =& v^\mu \partial_{v^{l_{n+1}}}(v^l(f_{l_1\cdots l_n})_l)\\
        =& 0. }
    \end{equation}

    Hence integrating along the ray from origin in $\mathbf{v}$ space we have that $v^\mu(f_{l_1\cdots l_n})_{l_{n+1}}$ at any point in $\mathbf{v}$ space is the same as at $\mathbf{v} = \mathbf{0}$. 
    
    Thus $v^\mu(f_{l_1\cdots l_n})_{l_{n+1}} = 0$ always and hence $(f_{l_1\cdots l_n})_{l_{n+1}} = 0$ always. This gives
    
    \begin{equation}
        u^{ij}\partial_{l_1\cdots l_nl_{n+1}}(g_{ij} - b_{ij}) = 0, 
    \end{equation}
    
    which finishes the induction. 
\end{proof}

We have two corollaries. Both still assume that $H$ is regular and generates a velocity-independent force. 

\begin{cor} \label{crucial aq}
    If at an $\mathbf{x_0}\in D$, $g_{ij}$ is constant as a function of $\mathbf{v}$, then $H$ is equivalent (via permissible canonical transformations) to an anisotropic quadratic Hamiltonian (see \cite{Berry_2015}). 
\end{cor}

\begin{proof}
    We can use a permissible canonical transformation to assume that $\mathbf{v} = 0$ iff $\mathbf{p} = 0$. 

    Let $g_{ij} = c_{ij}$ at $\mathbf{x_0}$. 
    
    By theorem \ref{crucial}, we have $\mathcal{A} = \mathcal{A}(\mathbf{x_0}) = \{c_{ij}\}$. 

    Thus $g_{ij} = c_{ij}$ always. Hence $g^{ij} = c^{ij}$ always, where $c^{ij}$ is the inverse of $c_{ij}$. 

    As $\partial^i H = 0$ at $\mathbf{p} = 0$, we must have $H(\mathbf{x}, \mathbf{p}) = \frac{1}{2}c^{ij}p_ip_j + H(\mathbf{x}, 0)$, an anisotropic quadratic Hamiltonian. 
\end{proof}

\begin{cor} \label{crucial force}
    $T^k_j(\mathbf{y})g_{ik}(\mathbf{x}, \mathbf{v})$ is symmetric for any $\mathbf{x}, \mathbf{y}\in D$, $\mathbf{v}\in \mathbb{R}^n$. 
\end{cor}

\begin{proof}
    We have $T^k_j(\mathbf{y})g_{ik}(\mathbf{y}, \mathbf{v})$ symmetric, hence for any $c_{ij}\in\mathcal{A}(\mathbf{y}) = \mathcal{A}$, $T^k_j(\mathbf{y})c_{ik}$ is symmetric. 
\end{proof}

Already we can see the special nature of the `anisotropic quadratic' Hamiltonians introduced in \cite{Berry_2015}. For these, $g$ is constant, allowing for the widest range of possibilities for $T$ and the dynamics. 

\section{The 1D Problem} \label{1D prob sec}
\begin{thm} \label{1D Hamiltonian}
    Let $H$ be a Hamiltonian defined on $(x, p)\in D\times \mathbb{R}$, where $D\subseteq\mathbb{R}$ is open and connected. If $H$ is regular (cf.~subsection \ref{regularity}, this is a technical assumption) and generates a velocity-independent force, then after a permissible canonical transformation (cf.~subsection \ref{canon trans sec}), either
    \begin{enumerate}
        \item there exists a function $f$ defined on the range of $H$, which is analytic on the interior of its domain, such that 
        \begin{equation}\label{1D global equation}
            \partial_{pp} H = f(H), 
        \end{equation}or
        \item $H$ is a function of $p$. (This generates a vanishing force.)
    \end{enumerate}
    And any regular Hamiltonian $H$ satisfying one of the above indeed generates a velocity-independent force. 
\end{thm}

We now prove this theorem. $D$ must be an (potentially infinite) open interval in this case. As we are working in 1D, we shall simply use $x, p, v, F$ to denote position, momentum, velocity and force. We use $\partial_p$ to denote momentum derivatives. 

In this 1D problem, the velocity-independence of the force is equivalent to (\ref{multiD fund2}), which becomes

\begin{equation} \label{1D fund2}
    \{\partial_{pp} H, H\} = 0. 
\end{equation}

As (\ref{1D fund2}) holds for both cases in theorem \ref{1D Hamiltonian}, its last assertion, that a regular Hamiltonian satisfying either i or ii generates a velocity-independent force, follows. 

To prove that any regular Hamiltonian generating a velocity-independent force satisfies either i or ii, we split into two cases based on whether the force generated is (universally) vanishing. They corresponding to the cases i and ii of theorem \ref{1D Hamiltonian}, respectively. 

\subsection{Case i: Non-vanishing Force}

It follows from (\ref{1D fund2}) that locally, away from stationary points of $H$, we have

\begin{equation} \label{1D equation}
    \partial_{pp} H = f(H)
\end{equation}

for some analytic function $f$ of $H$. 

The only obstacles to analytic continuation are contours of $H$ consisting solely of stationary points of $H$. Note that stationary points of $H$ are exactly points where $v = 0$ and $F = 0$. Hence if the force is not universally vanishing, we may proceed to extend (\ref{1D equation}) across $v = 0$ to the entire phase space. This gives (\ref{1D global equation}) for a suitable choice of function $f$ whose domain is the range of $H$, which in addition is analytic in the interior of its domain, as required. 

Note that for $H$ to be regular, the Hessian $\partial_{pp} H$ (of the $p - v$ bijection) and hence $f(H)$ must be always positive or always negative. 

As (\ref{1D global equation}) is a second order differential equation of $H$ as a function of $p$ for each $x$, we have two integration constants, which are in this context functions of $x$. However, one of these is redundant as it can be eliminated by a canonical transformation. Indeed, we can apply a canonical transformation of type $2$ so that $v = 0$ and $p = 0$ coincide. This is the requirement that $\partial_p H = 0$ at $p = 0$. Thus after we determine a suitable $f$, the Hamiltonian is in general determined by one more function of space $x$. 

\subsubsection{Examples}

$f$ being constant gives the standard Hamiltonian

\begin{equation}
    H = \frac{f}{2}p^2 + U(x). 
\end{equation}

$f$ being the identity function gives 

\begin{equation}
    H = \cosh(p)U(x), 
\end{equation}

generating the force 

\begin{equation}
    F(x) = - U(x)U'(x). 
\end{equation} 

Note that in this example we must have $U(x)$ being always positive/negative for there to be an analytic diffeomorphism between $\mathbf{p}$ and $\mathbf{v}$. A specific instance of this example would be 

\begin{equation}
    H = \cosh(p)\sqrt{x^2 + a^2}
\end{equation} 

for any $a\neq 0$ generating the simple harmonic motion 

\begin{equation}
    F = -x, 
\end{equation}

where the $a\neq 0$ requirement ensures that $U(x)$ is always positive. 

\subsection{Case ii: Vanishing Force}

By the flow in state space $(x, v)$ generated by the dynamics, we have that the conserved quantities $H$ and $\partial_{pp} H$ are functions of $v$ when $v\neq 0$. By continuity we know that $H$ and $\partial_{pp} H$ are always functions of $v$ only. 

Using a canonical transformation of type 2, we assume that $p = 0$ when $v = 0$. As the Hessian of the $v - p$ map $\frac{1}{\partial_{pp} H}$ is a function of $v$ only, we have that $p$ is a function of $v$. That is, the $v - p$ bijection is independent of $x$. Thus $H$, which is a function of $v$, is a function of $p$. 

\subsubsection{Examples}

The standard Hamiltonian of a 1D free particle 

\begin{equation}
    H = \frac{1}{2}p^2
\end{equation}

is an example which also satisfies (\ref{1D global equation}). 

Though in general (\ref{1D global equation}) would not be satisfied in this case, such as for the Hamiltonian 

\begin{equation}
    H = \rme^p + p^2, 
\end{equation}

hence a separate discussion is indeed warranted. 

\section{The 2D Problem} \label{2D prob sec}

We will prove the following theorem for the 2D problem, completely classifying all possible Hamiltonians, apart from ones generating isotropic simple harmonic motion. Here an isotropic simple harmonic motion refers to a force given by $\mathbf{F} = a \mathbf{x} + \mathbf{b}$ for a fixed scalar $a$ and vector $\mathbf{b}$. 

To avoid confusion with squares of these quantities, here we will use $x, y, v_x, v_y, p_x, p_y, F_x, F_y$ to denote the components $x^1, x^2, v^1, v^2, p_1, p_2, F^1, F^2$ of position, velocity, momentum and force respectively. 

\begin{thm} \label{2D Hamiltonian}
    Let $H$ be a Hamiltonian defined on $(\mathbf{x}, \mathbf{p})\in D\times \mathbb{R}^2$, where $D\subseteq\mathbb{R}^2$ is open and connected. Let $H$ be regular (cf.~subsection \ref{regularity}, this is a technical assumption) and generate a velocity-independent force, which is not an isotropic simple harmonic motion. 
    
    Then $H$ is of one of the following three forms after a permissible canonical transformation (cf.~subsection \ref{canon trans sec}). 
    \begin{enumerate}
        \item `Anisotropic Quadratic' (introduced in \cite{Berry_2015}), 
        \begin{equation}
            H = \frac{1}{2} M^{ij} p_ip_j + U(\mathbf{x}), 
        \end{equation}
        where $M$ is a fixed invertible $2\times 2$ real matrix, and $U: D \rightarrow \mathbb{R}$ is analytic. 
        \item `Separable', 
        \begin{equation}
            H = H_1(x, p_x) + H_2(y, p_y), 
        \end{equation}
        where $H_1$ and $H_2$ are regular 1D Hamiltonians generating velocity-independent forces. (cf.~section \ref{1D prob sec})
        \item `Seesaw', here we have $2$ sub-cases: 
        \begin{enumerate}
            \item $H$ is of the form
            \begin{equation}
                \fl\eqalign{H = y \partial'_{x} \eta(x, p_y) + u(x) \partial'_{x} \eta(x, p_y) + p_x \partial_{p_y} \eta (x, p_y) + G(\eta(x, p_y))\cr
                - \partial_{p_y} \eta(x, p_y)  \int_0^{p_y} G'(\eta(x, p))\, \rmd p, }
            \end{equation}
            where $\eta$ is a regular Hamiltonian generating a somewhere non-vanishing velocity-independent force, and $G$ is an analytic function. $u$ is allowed to have singularities so long as $u(x) \partial'_{x} \eta(x, p_y)$ is analytic. 
            \item $H$ is of the form
            \begin{equation}
                H = \left(p_x + c(x) + d(x)\int^{p_y}\frac{1}{r(p)^2}\, \rmd p\right) r(p_y) + u(p_y), 
            \end{equation}
            where $r$ is an analytic diffeomorphism $\mathbb{R}\rightarrow\mathbb{R}$, and we take the Cauchy Principle Value of the integral when $r$ vanishes. $c$, $d$, $u$ are analytic. 
        \end{enumerate}
    \end{enumerate}
    And all the Hamiltonians $H$ listed above are indeed regular Hamiltonians generating a velocity-independent force. 
\end{thm}

Unlike the first two cases, in the third case (Seesaw), the Hamiltonian is always unbounded above and below (even at each fixed $\mathbf{x}$), thus may not be physical. The split into two sub-cases for the Seesaw Hamiltonian is similar to that in the 1D problem. 

A note here on the three types of Hamiltonians: 
\begin{enumerate}
    \item `Anisotropic quadratic' is the name given in \cite{Berry_2015} to these Hamiltonians with a not-necessarily-isotropic quadratic kinetic term $\frac{1}{2} M^{ij} p_ip_j$. 
    \item A `separable' Hamiltonian is one where the motion decouples both classically and quantum mechanically. 
    \item The name `Seesaw' pays homage to the Seesaw Mechanism \cite{SeesawMinkowski, SeesawTsutomu} for neutrino mass, to which the Hamiltonian and dynamics of the third case bear resemblance in two aspects to be discussed in the proof. 
\end{enumerate}

We will obtain also the characterisation of velocity-independent Hamiltonian forces, practically as a corollary. 

\begin{thm} \label{force thm}
    If a velocity-independent force $F$ is generated by a regular Hamiltonian $H$ defined on $(\mathbf{x}, \mathbf{p})\in D\times \mathbb{R}^2$, then $F$ can be generated by an anisotropic quadratic Hamiltonian $\tilde{H} = \frac{1}{2} M^{ij} p_ip_j + U(\mathbf{x})$, where $M$ is a fixed invertible $2\times 2$ real matrix, and $U: D\rightarrow \mathbb{R}$ is analytic. 
\end{thm}

Note that as a canonical transformation modifies the force by a linear transformation ($F^i$ transforms as a vector), and the ability to be generated by an anisotropic quadratic Hamiltonian is invariant under linear transformations, it suffices to prove this after any suitable canonical transformations are made. 

We shall now prove both theorems as well as giving examples of separable and Seesaw Hamiltonians by a case discussion using the general results established in section \ref{highD}. 

Note that $T$ is a $(1, 1)$ tensor, hence one could examine its eigenvalues. Indeed we shall do so here. One sees that the techniques developed can be applied to higher dimensions, though a more careful analysis is needed due to more possible Jordan Normal Forms. 

We split into 4 cases based on the eigen decomposition of $T$. 

\begin{enumerate}
    \item \emph{Somewhere} $T$ has a pair of complex (non-real) conjugate eigenvalues
    \item \emph{Somewhere} $T$ has distinct real eigenvalues
    \item \emph{Everywhere} $T$ has only one (real) eigenvalue, and \emph{somewhere} $T$ is not isotropic
    \item \emph{Everywhere} $T$ is isotropic
\end{enumerate}

The generic case in case 1 will lead to a contradiction. The remaining special case gives an anisotropic quadratic Hamiltonian. 

Case 2 leads to a separable Hamiltonian or an anisotropic quadratic Hamiltonian. 

Case 3 leads to a Seesaw Hamiltonian or an anisotropic quadratic Hamiltonian. 

Case 4 produces an isotropic simple harmonic motion. 

\subsection{Distinct Eigenvalues}

We consider right/lower eigenvectors, that is we consider the eigenvector $e_i$ with eigenvalue $\lambda$ when $\lambda e_i = T_i^je_j$. 

What we discuss here is common to cases 1 and 2. We assume that somewhere, say at $\mathbf{y}$, $T$ has distinct eigenvalues $\lambda^\alpha$ with eigenvectors $(e^\alpha)_i$, which we now fix. Note the lack of implicit summation over eigen-indices; they are not to be viewed as tensorial. 

As the two eigenvectors form a basis, we can expand 

\begin{equation}
    g_{ij} = \sum_{\alpha\beta}A^{\alpha\beta}(e^\alpha)_i(e^\beta)_j, 
\end{equation}

where $A^{\alpha\beta}$ is symmetric and potentially complex if the eigenvalues are. 

Now corollary \ref{crucial force} gives that 

\begin{equation}
    \eqalign{T^k_j(\mathbf{y})g_{ik}(\mathbf{x}, \mathbf{v}) &= \sum_{\alpha\beta}T^k_j(\mathbf{y})A^{\alpha\beta}(\mathbf{x}, \mathbf{v})(e^\alpha)_i(e^\beta)_k \cr
    &= \sum_{\alpha\beta}\lambda^\beta A^{\alpha\beta}(\mathbf{x}, \mathbf{v})(e^\alpha)_i(e^\beta)_j}
\end{equation}
is symmetric in $i, j$. Thus $\lambda^\beta A^{\alpha\beta}$ is symmetric in $\alpha, \beta$ and $A$ must be diagonal. Therefore

\begin{equation}
    g_{ij} = \sum_{\alpha}A^{\alpha\alpha}(e^\alpha)_i(e^\alpha)_j. 
\end{equation}

Property 1 (\ref{g basic 1}) and property 4 (\ref{g basic 4}) combined with this give
\numparts
\begin{eqnarray}
    \nabla_{\mathbf{v}}A^{\alpha\alpha} \parallel (\mathbf{e}^\alpha), \\
    \nabla_{\mathbf{x}}A^{\alpha\alpha} \parallel (\mathbf{e}^\alpha)
\end{eqnarray}
\endnumparts

respectively. 

\subsection{Case 1}
Case 1: Somewhere, say at $y\in D$, $T$ has a pair of complex (non-real) conjugate eigenvalues. 

Let 

\begin{equation}
    p_i = \sum_{\alpha}B^\alpha (e^\alpha)_i. 
\end{equation}

As $g_{ij}$ is the Hessian of the $\mathbf{v} - \mathbf{p}$ map, we have

\begin{equation}
    \rmd p_i = g_{ij}\rmd v^j. 
\end{equation}

Noting that $(e^\alpha)_i$ is fixed, it follows that

\begin{equation}
    \sum_\alpha (e^\alpha)_i \rmd B^\alpha = \sum_\alpha A^{\alpha\alpha}(e^\alpha)_i\rmd((e^\alpha)_jv^j). 
\end{equation}

Hence we have 

\begin{equation}
    \rmd B^\alpha = A^{\alpha\alpha}\rmd((e^\alpha)_jv^j). 
\end{equation}

Thus $B^\alpha$ is a holomorphic function of $v^i(e^\alpha)_i$ with derivative $A^{\alpha\alpha}$. 

At each fixed position $\mathbf{x}$ and for each choice of $\alpha$ we have bijections: 
\begin{enumerate}
    \item $\mathbb{C}\rightarrow\mathbb{R}^2$, $B^\alpha\mapsto \mathbf{p}$
    \item $\mathbb{R}^2\rightarrow\mathbb{R}^2$, $\mathbf{p}\mapsto \mathbf{v}$
    \item $\mathbb{R}^2\rightarrow\mathbb{C}$, $\mathbf{v}\mapsto v^i(e^\alpha)_i$
\end{enumerate}

Thus $\mathbb{C}\rightarrow\mathbb{C}$, $v^i(e^\alpha)_i\mapsto B^\alpha$ is a holomorphic bijection, and therefore linear, hence $g_{ij}$ is constant as a function of $\mathbf{v}$ at $\mathbf{x}$, hence by corollary \ref{crucial aq}, $H$ is equivalent to an anisotropic quadratic Hamiltonian. 

\subsection{Case 2}
Case 2: Somewhere, say at $y\in D$, $T$ has distinct real eigenvalues. 

By a change of basis/canonical transformation of type 1 we may assume that the (necessarily real) eigenvectors are 

\begin{equation}
    \mathbf{(e^1)} = \left(\matrix{
    1 \cr
    0
}\right), \qquad\mathbf{(e^2)} = \left(\matrix{
    0 \cr
    1
}\right). 
\end{equation}

Thus in this frame, $T$ is diagonal. 

From here onwards, we stop respecting the tensorial nature, and as such we suspend implicit summation. As proven before, we may assume by a permissible canonical transformation that $\mathbf{v} = \mathbf{0}$ iff $\mathbf{p} = \mathbf{0}$. 

From the earlier discussion common to case 1 and 2, we have that $g$ is diagonal and $g_{\alpha\alpha}$ is a function of $x^\alpha$ and $v^\alpha$. As

\begin{equation}
    \rmd p_\alpha = g_{\alpha\alpha}\rmd v^\alpha
\end{equation}

and $p_\alpha = 0$ when $v^\alpha = 0$, we have that $p_\alpha$ is a function of $x^\alpha$ and $v^\alpha$ only. 

Hence for each choice of $\alpha = 1, 2$, we have a bijection between $p_\alpha$ and $v^\alpha$ which is determined by $x^\alpha$. Thus $g_{\alpha\alpha}$ is a function of $x^\alpha$ and $p_\alpha$ and so is $g^{\alpha\alpha} = \frac{1}{g_{\alpha\alpha}}$. 

We have 

\begin{equation}
    H(\mathbf{x}, \mathbf{p}) = \sum_\alpha I_\alpha(x^\alpha, p_\alpha) + H(\mathbf{x}, \mathbf{0}), 
\end{equation}

where $\partial_{p_\alpha p_\alpha}I_\alpha = g^{\alpha\alpha}$ and $I_\alpha(x^\alpha, 0) = \partial_{p_\alpha}I_\alpha(x^\alpha, 0) = 0$. 

If $T$ is not always diagonal, then at some $\mathbf{z}$, we have $T_1^2(\mathbf{z})$ and $T_2^1(\mathbf{z})$ are not both $0$. Corollary \ref{crucial force} implies that 

\begin{equation}
    T_1^2(\mathbf{z})g_{22} = T_2^1(\mathbf{z})g_{11}. 
\end{equation} 

As $g$ is never singular, both $g_{11}$ and $g_{22}$ are never $0$, hence both $T_1^2(\mathbf{z})$ and $T_2^1(\mathbf{z})$ are not $0$ and therefore $g^{11}$ and $g^{22}$ are proportional. But as they are functions of different variables, this implies that they are both constant and we have an anisotropic quadratic Hamiltonian. Thus we assume now that $T$ is always diagonal. In other words, $F^\alpha$ is a function of only $x^\alpha$. 

Applying property 3 (\ref{g basic 3}) at $\mathbf{p} = \mathbf{0}$ which coincides with $\mathbf{v} = \mathbf{0}$, we have 

\begin{equation}
    \partial_{\alpha} H(\mathbf{x}, \mathbf{0}) = - g_{\alpha\alpha}F^\alpha
\end{equation}

is a function of $x^\alpha$. Hence we may separate 

\begin{equation}
    H(\mathbf{x}, \mathbf{0}) = \sum_{\alpha}u_\alpha(x^\alpha). 
\end{equation} 

Letting 

\begin{equation}
    H_\alpha(x^\alpha, p_\alpha) = I_\alpha(x^\alpha, p_\alpha) + u_\alpha(x^\alpha), 
\end{equation}

we have 

\begin{equation}
    H = \sum_\alpha H_\alpha(x^\alpha, p_\alpha) = H_1(x, p_x) + H_2(y, p_y)
\end{equation} 

as required and we can see that $H_\alpha$ must be regular 1D Hamiltonians generating velocity-independent forces. Hence we have proven that $H$ is separable. As any separable force is conservative, it may be generated by an anisotropic quadratic Hamiltonian. 

\subsubsection{An Example}

One may consider the separable Hamiltonian

\begin{equation}
    H = \frac{1}{2}p_x^2 + \frac{1}{2}x^2 + \cosh(p_y)(y^2 + 1)
\end{equation}

generating the force 

\begin{equation}
    \mathbf{F} = \left(\matrix{
        -x \cr
        -2y^3 - 2y
    }\right). 
\end{equation}

\subsection{Case 3}

Case 3: Everywhere $T$ has only one (real) eigenvalue, and somewhere, say at $\mathbf{y}$, $T$ is not isotropic. 

Again we stop respecting the tensorial nature from here onwards and suspend implicit summation. By a change of basis/canonical transformation of type 1, put $T$ into Jordan Normal Form at $\mathbf{y}$, so $T_2^1 = 0$, $T_1^2 \neq 0$. This gives that $g_{22} = 0$ always. 

If it is not the case that $T_2^1 = 0$ always, then we must have $g_{11} = 0$ too, which gives by Properties 1 and 4 (\ref{g basic 1}, \ref{g basic 4}) that $g_{12}$ is a constant and hence $H$ is equivalent to an anisotropic quadratic Hamiltonian by corollary \ref{crucial aq}. Thus we assume that $T_2^1 = 0$ always, which gives that 
\begin{eqnarray}
    F_x = a(x)\\
    F_y = b(x) + a'(x) y
\end{eqnarray}

for some functions $a, b$ of $x$. Note that $F_x$ is a function of $x$ only. For an analytic function $u: \mathbb{R}\rightarrow\mathbb{R}$, we denote by $\int u$ any fixed anti-derivative of $u$. For the sake of the classification of forces (theorem \ref{force thm}), the force may be generated by the anisotropic quadratic Hamiltonian 

\begin{equation}
    H = p_xp_y - y a(x) - \int b (x). 
\end{equation}

Properties 1 and 4 (\ref{g basic 1}, \ref{g basic 4}) give that 
\begin{eqnarray}
    g_{12} = h(x, v_x), \\
    g_{11} = f(x, v_x) + y \partial_x h + v_y \partial_{v_x} h
\end{eqnarray} 

for some functions $h, f$ of $x, v_x$. Note that $g_{12}$ is a function of $x, v_x$ only. 

Property 2 (\ref{g basic 2}) gives

\begin{equation} \label{h before characteristics}
    v_x\partial_x h + a(x)\partial_{v_x} h = 0, 
\end{equation}
\begin{equation} \label{f before characteristics}
    v_x\partial_x f + a(x)\partial_{v_x} f + b(x)\partial_{v_x} h = 0. 
\end{equation}

Similar to the 1D problem (section \ref{1D prob sec}), we further split into two subcases based on whether the force $F_x = a$ vanishes. 

\subsubsection{Sub-case i}

Here we assume that $a$ doesn't universally vanish. 

Solving the differential equations (\ref{h before characteristics}) and (\ref{f before characteristics}) by method of characteristics, we have 

\begin{equation} \label{h after characteristics}
    h = \psi\left(\frac{1}{2}v_x^2 - I(x)\right), 
\end{equation}
\begin{equation} \label{f after characteristics}
    f = - J(x) \psi'\left(\frac{1}{2}v_x^2 - I(x)\right) + \theta\left(\frac{1}{2}v_x^2 - I(x)\right). 
\end{equation}

where 

\begin{equation}
    I = \int a, \qquad J = \int b. 
\end{equation}

A technical note here: we need the fact that $a$ doesn't universally vanish, or that $I$ is not a constant function, so that characteristics would cross $v_x = 0$ in some regions, from which we can extend by analyticity. This is similar to the treatment in section \ref{1D prob sec} for the 1D problem. As we shall see below in sub-case ii, without this assumption, $h$, $f$ don't have to be symmetric under 

\begin{equation}
    v_x \mapsto -v_x. 
\end{equation}

Using property 3 (\ref{g basic 3'}, \ref{g basic 3}), up to an additive constant, we have

\begin{equation}
    H = (v_xv_y - J(x) - y a(x)) \psi\left(\frac{1}{2}v_x^2 - I(x)\right) + \phi\left(\frac{1}{2}v_x^2 - I(x)\right), 
\end{equation}

where 

\begin{equation}
    \phi = \int \theta. 
\end{equation}

It remains to convert this from state space ($\mathbf{x}, \mathbf{v}$) into phase space ($\mathbf{x}, \mathbf{p}$). The relation between $x$, $p_y$, $v_x$ are exactly those of position, momentum and velocity in a 1D Hamiltonian generating the velocity-independent force $a(x)$. In fact this 1D Hamiltonian is given by $\eta = \int \psi$ as a function of $x$ and $p_y$. As $\psi\left(\frac{1}{2}v_x^2 - I(x)\right) = g_{12}$ is always positive or always negative, $\eta$ is a strictly monotonic function of $\frac{1}{2}v_x^2 - I(x)$. Hence we can write 

\begin{equation}
    \phi = G(\eta)
\end{equation} 

for some function $G$. Converting H to a function of momentum we have that

\begin{equation} \label{type a seesaw}
    \fl\eqalign{H = y \partial_x \eta(x, p_y) + \frac{J(x)}{a(x)} \partial_x \eta(x, p_y) + p_x \partial_{p_y} \eta (x, p_y) + G(\eta(x, p_y)) \cr
    - \partial_{p_y} \eta(x, p_y)  \int_0^{p_y} G'(\eta(x, p))\, \rmd p. }
\end{equation}

Hence with 

\begin{equation}
    u(x) = \frac{J(x)}{a(x)}, 
\end{equation}

$H$ is a type (a) Seesaw Hamiltonian. Note that intriguingly, the two terms involving $G$ do not contribute to the dynamics of the system classically. 

\subsubsection{An Example}
In the general form (\ref{type a seesaw}), we may take $G$ to be vanishing and $\eta(x, p) = \cosh(p)\sqrt{x^2 + 1}$ to give the type (a) Seesaw Hamiltonian

\begin{equation}
    H = (y + u(x))\cosh(p_y)\frac{x}{\sqrt{x^2 + 1}} + p_x\sinh(p_y)\sqrt{x^2 + 1}, 
\end{equation}

which generates the force 

\begin{equation}
    \mathbf{F} = \left(\matrix{
        -x \cr
        -y - u(x) - x u'(x)
    }\right). 
\end{equation}

More specifically, one could consider the example 

\begin{equation}
    H = (y + x)\cosh(p_y)\frac{x}{\sqrt{x^2 + 1}} + p_x\sinh(p_y)\sqrt{x^2 + 1}, 
\end{equation}

which generates the force 

\begin{equation}
    \mathbf{F} = \left(\matrix{
        -x \cr
        -y - 2x
    }\right). 
\end{equation}

\subsubsection{Sub-case ii}

This is the case where $a$ vanishes. Solving (\ref{h before characteristics}) now gives that 

\begin{equation}
    \partial_x(h) = 0
\end{equation}

and hence 

\begin{equation}
    \partial_y(g_{ij}) = 0. 
\end{equation} 

Therefore the bijection $\mathbf{v} - \mathbf{p}$ is dependent on $x$ only. 

Property 3 (\ref{g basic 3}) gives that 

\begin{equation}
    \partial_y H = 0. 
\end{equation} 

In other words, $H$ is a function of $x, v_x, v_y$ and hence a function of $x, p_x, p_y$. Thus we must have

\begin{equation}
    H = p_x r(x, p_y) + t(x, p_y). 
\end{equation}

This generates a force whose $x$ component is

\begin{equation}
    F_x = r \partial_x r. 
\end{equation}

As this must vanish by assumption, we have that $r$ is a function of $p_y$ only. Note that as 

\begin{equation}
    v_x = r(p_y), 
\end{equation}

$r$ must be an analytic diffeomorphism. The $y$ component of the force generated is now

\begin{equation}
    F_y = - r'(p_y)\partial_x t(x, p_y) + r(p_y) \partial_{x p_y} t(x, p_y)
\end{equation} 

and its independence of $p_y$ gives

\begin{equation}
    - r''(p_y)\partial_x t(x, p_y) + r(p_y) \partial_{x p_y p_y} t(x, p_y) = 0. 
\end{equation}

Integrating with respect to $x$ gives

\begin{equation}
    - r''(p_y) (t(x, p_y) - t(0, p_y)) + r(p_y) \partial_{p_y p_y} (t(x, p_y) - t(0, p_y)) = 0, 
\end{equation}

which (for each $x$) is a second order \textit{linear} differential equation for $t(x, p_y) - t(0, p_y)$ as a function of $p_y$. We solve for this by noting that 

\begin{equation}
    t(x, p_y) - t(0, p_y) = r(p_y)
\end{equation}

is a solution and performing reduction of order. This gives

\begin{equation}
    t(x, p_y) - t(0, p_y) = r(p_y) \left(c(x) + d(x)\int^{p_y} \frac{dp}{r(p)^2}\right), 
\end{equation}

which produces the desired form of a type (b) Seesaw Hamiltonian if we take 

\begin{equation}
    u(p_y) = t(0, p_y). 
\end{equation} 

\subsubsection{An Example}
We may take $r$ to be the identity to give the type (b) Seesaw Hamiltonian

\begin{equation}
    H = p_xp_y + d(x)p_y + a(p_y) - c(x), 
\end{equation}

which generates the force 

\begin{equation}
    \mathbf{F} = \left(\matrix{
        0 \cr
        c'(x)
    }\right). 
\end{equation}

More specifically, one could consider the example 

\begin{equation}
    H = p_xp_y - \rme^{p_y} + x^2, 
\end{equation} 

which generates the force 

\begin{equation}
    \mathbf{F} = \left(\matrix{
        0 \cr
        - 2x
    }\right). 
\end{equation}

One could see that what we derived in sub-case i, such as (\ref{f after characteristics}), does not hold in this example here, and hence it is really necessary to split into two sub-cases here. 

\subsubsection{The Name `Seesaw'}

There are two aspects of the Seesaw Hamiltonian reminiscent of the Seesaw Mechanism \cite{SeesawMinkowski, SeesawMassMatrix1, SeesawTsutomu, SeesawMassMatrix2} proposed as a potential explanation for small neutrino mass. 

To exploit the full power of the analogy, note that $x$, $y$ are not necessarily to be interpreted as two position coordinates of the same particle; instead we can interpret them as (generalized) positions of two coupled particles. 

\begin{enumerate}
	\item Mass Tensor: 
	The analogy of the mass tensor in this case is $g_{ij}$, which consists of a vanishing diagonal element as in the Seesaw Mechanism. Furthering the analogy, a mass tensor analogous to $g_{ij}$ here would generate a heavy particle near the $x$ direction and a light particle near the $y$ direction \cite{SeesawMassMatrix1, SeesawMassMatrix2}. 
	\item Dynamics: 
	The dynamics in the $x$ direction is completely independent of $y$ position, whilst the dynamics in the $y$ direction is driven by $x$ position, as is expected in the limiting regime of a heavy particle in the $x$ direction coupled to a light particle in the $y$ direction. 
\end{enumerate}

\subsection{Case 4}

If $T_i^j$ is always isotropic, let 

\begin{equation}
    T_i^j = a(x)\delta_i^j. 
\end{equation}

As 

\begin{equation}
    T_i^j = \partial_i (F^j), 
\end{equation} 

we have that 

\begin{equation}
    \partial_{ik} F^j = \partial_k(T_i^j) = \partial_k (a)\delta_i^j
\end{equation}

is symmetric in $i$, $k$. Thus 

\begin{equation}
    \partial_i (a)\delta_k^j = \partial_k (a)\delta_i^j. 
\end{equation}

Contracting $i$ and $j$ gives 

\begin{equation}
    \partial_k (a) = 2\partial_k (a), 
\end{equation}

implying that $a$ is constant. Hence $\mathbf{F} = a \mathbf{x} + \mathbf{b}$ for some $\mathbf{b}$ and we have isotropic simple harmonic motion, which for the sake of the classification of forces (theorem \ref{force thm}) is conservative and thus can be generated by an anisotropic quadratic Hamiltonian. 

\section{Pseudo Conservative Forces} \label{pseudo conservative}

A force $\mathbf{F}$ is conservative if it can be written as 
\begin{equation} \label{conservative}
    \mathbf{F} = -\nabla U
\end{equation}
for a scalar field $U$. However, written in index notation, we have a mismatch between the two sides $F^i$ and $\partial_i U$. This motivates the following definition. 

\begin{df}
    $\mathbf{F}$ is pseudo conservative if $F^i = -M^{ij}\partial_j U$ for a fixed invertible symmetric $n\times n$ matrix $M$. 
\end{df}

From this perspective it may be seen that the conservative nature of many forces is underpinned by fixing a positive definite spatial metric $M^{ij}$. 

Noting that the anisotropic quadratic Hamiltonian
\begin{equation}
    H = \frac{1}{2} M^{ij} p_ip_j + U(\mathbf{x})
\end{equation}
generates the force
\begin{equation}
    F^i = -M^{ij}\partial_j U, 
\end{equation}
pseudo conservative forces are precisely the forces which may be generated by an anisotropic quadratic Hamiltonian. 

Theorem \ref{force thm} is now equivalent to the following corollary. 

\begin{cor}
    In two spatial dimensions (over any open and connected spatial domain), any velocity-independent force generated by a regular Hamiltonian is pseudo conservative. 
\end{cor}

We similarly extend the concept of a curl-free force. 

\begin{df}
    $\mathbf{F}$ is pseudo curl-free if $F^i = -M^{ij}u_j$ for a fixed invertible symmetric $n\times n$ matrix $M$ and a curl-free vector field $\mathbf{u}$, that is, $\partial_iu_j$ is symmetric. 
\end{df}

Note that pseudo conservative forces are pseudo curl-free, and the converse holds whenever $H^1_{dR}(D) = 0$, such as when $D = \mathbb{R}^n$. The following result in general dimensions follows from corollary \ref{crucial force}. 

\begin{cor}
    In any open and connected spatial domain $D$ in any dimension, any velocity-independent force generated by a regular Hamiltonian is pseudo curl-free. If further we have $H^1_{dR}(D) = 0$, then the force is pseudo conservative. 
\end{cor}

\section{Conclusion}

We developed strong results about regular Hamiltonians generating velocity-independent forces in general dimensions (section \ref{highD}).  The method of working in state space ($\mathbf{x}, \mathbf{v}$) without translating to Lagrangian formalism proved helpful in establishing these results. 

We applied these general results to completely classify regular (cf.~subsection \ref{regularity}) Hamiltonians in one spatial dimension generating velocity-independent forces (section \ref{1D prob sec}, theorem \ref{1D Hamiltonian}), and all such Hamiltonians in two spatial dimensions apart from ones generating an isotropic simple harmonic motion (section \ref{2D prob sec}, theorem \ref{2D Hamiltonian}). All such velocity-independent forces generated by a regular Hamiltonian in two spatial dimensions were also classified (theorem \ref{force thm}), and their relation to conservative forces and pseudo conservatism was explored in section \ref{pseudo conservative}. 

Some further directions arising out of this discussion include: 

\begin{enumerate}
    \item A study of Hamiltonians which generate isotropic simple harmonic motion. 
    \item Extension of this scheme of characterisation to higher and ideally general dimensions. 
    \item A better understanding of the quantisation properties of the new classes of Hamiltonians proposed here. 
\end{enumerate}

\ack

We thank Sir Michael V Berry for his thoughtful insights and helpful suggestions. 

\appendix
\section*{Appendix}
\setcounter{section}{1}

We consider here time-dependent Hamiltonians
\begin{equation}
    H = H(\mathbf{x}, \mathbf{p}, t), 
\end{equation}
generating the velocity-independent force
\begin{equation}
    \mathbf{F} = \mathbf{F}(\mathbf{x}, t), 
\end{equation}
which is now allowed to be time-dependent. The conditions for the velocity-independence of the force now take the form

\begin{equation} \label{timedep multiD fund1}
    \{\partial^i H, \partial^jH\} = 0,
\end{equation}

and

\begin{equation} \label{timedep multiD fund2}
    \partial_t\left(\partial^{ij}H\right) = \{\partial^{ij} H, H\}, 
\end{equation}

generalising (\ref{multiD fund1}), (\ref{multiD fund2}) of the time-independent case. Condition (\ref{timedep multiD fund2}) states that $\partial^{ij}H$ is a conserved constant of motion. Note that condition (\ref{timedep multiD fund1}) is imposed upon each temporal slice $H(\cdot, t)$ of the Hamiltonian, whilst condition (\ref{timedep multiD fund2}) prescribes the time evolution of the Hamiltonian. 

In one spatial dimension, condition (\ref{timedep multiD fund1}) is always satisfied. Given any regular time-independent Hamiltonian $H_0(\mathbf{x}, \mathbf{p})$ and time $t_0$, condition (\ref{timedep multiD fund2}) always yields local solutions near $t = t_0$ with
\begin{equation}
    H(\mathbf{x}, \mathbf{p}, t_0) = H_0(\mathbf{x}, \mathbf{p}). 
\end{equation}
On the other hand, the existence of temporally global solutions may be worthy of further investigation. A global solution to (\ref{timedep multiD fund2}) might not always exist, such as in the case of a temporally-separable Hamiltonian
\begin{equation}
    H(\mathbf{x}, \mathbf{p}, t) = S(t)H_0(\mathbf{x}, \mathbf{p}),
\end{equation}
where the temporal scaling $S(t)$ must satisfy
\begin{equation}
    \dot{S}(t) \propto S(t)^2, 
\end{equation}
and hence diverges at finite $t$ for any time-dependent solution. 

In higher dimensions, condition (\ref{timedep multiD fund1}) becomes significant. The interplay between condition (\ref{timedep multiD fund1}) on slices at fixed times and the time-evolution given by condition (\ref{timedep multiD fund2}) does not yield clear local solutions, and is worthy of further investigation. 

\printbibliography[
title = {References}
]

\end{document}